\documentclass[conference]{IEEEtran}
\IEEEoverridecommandlockouts
\usepackage{cite}
\usepackage{amsmath,amssymb,amsfonts}
\usepackage{algorithmic}
\usepackage{graphicx}
\usepackage{textcomp}
\usepackage{xcolor}
\usepackage{color}
\usepackage{epsfig}
\usepackage{subfig}
\usepackage{mathrsfs}
\usepackage{booktabs}
\usepackage{graphicx}
\usepackage{epstopdf}
\usepackage{multirow}
\usepackage{amsmath,amssymb}
\usepackage{amsthm}
\usepackage{mathptmx}
\usepackage{cite}

\newtheorem{theorem}{Theorem}[section]

\newtheorem{lemma}{Lemma}[section]
\newtheorem{remark}{Remark}[section]

\def\BibTeX{{\rm B\kern-.05em{\sc i\kern-.025em b}\kern-.08em
    T\kern-.1667em\lower.7ex\hbox{E}\kern-.125emX}}

\begin{document}

\title{Identification and Analysis of Cascading Failures in Power Grids with Protective Actions \\
\thanks{Corresponding author: Gaoxi Xiao}
}

\author{\IEEEauthorblockN{Chao Zhai}
\IEEEauthorblockA{\textit{School of Electrical and Electronic Engineering} \\
\textit{Nanyang Technological University} \\
Singapore \\
zhaichao@ntu.edu.sg}
\and
\IEEEauthorblockN{Gaoxi Xiao, and Hehong Zhang}
\IEEEauthorblockA{\textit{School of Electrical and Electronic Engineering} \\
\textit{Institute of Catastrophe Risk Management} \\
\textit{Nanyang Technological University} \\
Singapore \\
egxxiao@ntu.edu.sg}
}

\maketitle

\begin{abstract}
This paper aims to identify and analyze the initial contingencies or disturbances that could lead to the worst-case cascading failures of power grids. An optimal control approach is proposed to determine the most disruptive disturbances on the branch of power transmission system by regarding the disturbances as the control inputs. Moreover, protective actions such as load shedding and generation dispatch are taken into account in a convex optimization framework to prevent the cascading outages of power grids. In theory, the necessary conditions for identifying the most disruptive disturbances are obtained by solving an integrated system of algebraic equations. Finally, numerical simulations are carried out to validate the proposed approach on the IEEE RTS 24 Bus System.
\end{abstract}


\section{Introduction}
In the past decades, the world has suffered from several major blackouts such as US-Canada Blackout in 2003 \cite{us04}, European Blackout in 2006 \cite{eu06}, India Blackout in 2012 and Brazil Blackout in 1999 \cite{lu06}. All the major blackouts have caused huge economic losses and affected millions of people. Due to the complexity of electrical power systems, it has been a great challenge to understand, analyze and identify the cascading blackouts in practical power grids.

According to the analysis of technological reports, power system blackouts normally go through five stages: precondition, initiating event, cascade events, final state and restoration \cite{lu06}. The precondition usually happens in the winter or summer peak time due to the excessive power demand. The initiating event (e.g., short-circuit, overload, protection hidden failure, etc.) triggers the chain reaction of branch outages, which starts the cascade events. During the cascade event, power systems may take protective actions such as load shedding and generation dispatch to prevent the cascading failure. If protective actions fail to prevent the further cascades, power systems may end up with the final state of cascading failures, which tends to result in the blackout. As a result, the recovery strategy has to be taken in order to restore the normal state of power grids.  To avoid the occurrence of cascading blackouts, it would be desirable to identify the initial malicious disturbances or contingencies before cascading failures, so that the precautions can be taken in advance to eliminate the risk of power system blackouts.

Some identification approaches have been developed to search for the critical branches or initial malicious disturbances that can cause the large-scale disruptions \cite{chen05,dav11,don08,bie10,roc11,epp12,zhai17,zhai17iw}. For instance, some methods are proposed to identify the collections of $n-k$ contingencies based on the event trees \cite{chen05}, line outage distribution factor \cite{dav11} and other optimization techniques \cite{don08,bie10,roc11}. Nevertheless, these optimization approaches are not efficient to identify the large collections of $n-k$ contingencies that result in cascading blackouts. To address this problem, a ``random chemistry" algorithm is designed with the relatively low computational complexity \cite{epp12}. In addition, an optimal control approach is proposed to identify the initial contingencies by treating these contingencies as the control inputs \cite{zhai17,zhai17iw}. The above optimal control approach is able to identify the continuous changes of branch impedance other than direct branch outages as the initial contingencies. Moreover, it can be extended to identify the dangerous fluctuation of injected power on buses caused by the generation of renewable energies or load variations \cite{zhai18}. It is demonstrated that the optimal control approach can effectively determine the worst-case cascades of power grids without protective actions.

As is well known, protective actions such as load shedding and generation dispatch play a key role in preventing cascading blackouts of power grids. This work aims to extend the optimal control approach to identify the malicious contingencies by taking into account protective actions during the cascades. The main challenge in theory is to incorporate protective actions in the framework of optimal control theory, since protective actions are taken according to solutions of a different optimization problem (linear or nonlinear programming). By converting both the optimal control problem and the optimization problem for protective actions into an integrated system of algebraic equations, we manage to formulate and solve the problem of identifying initial contingencies that could lead to the worst-case cascading failures in power grids endowed with protective actions.

The remainder of this paper is organized as follows. Section \ref{sec:prob} formulates the contingency identification problem of power grids involving protective actions in the framework of optimal control theory. Section \ref{sec:the} presents the theoretical analysis of the optimal control problem. Section \ref{sec:sim} provides simulation results to validate the proposed approach. Finally, the conclusion is drawn in Section \ref{sec:con}.

\section{Problem Formulation}\label{sec:prob}
This section presents the problem formulation of identifying the most disruptive disturbances on transmission lines of power grids with $n$ branches and $m$ buses. Figure \ref{flow} presents an illustration of identifying the worst-case cascades of power grids with protective actions in the framework of optimal control theory. When the initial contingency or disturbance is added on the branch of power grids, the branch impendence is changed. This leads to the redistribution of power flow on the branches. If the power flow on the branch exceeds its threshold, the branch is severed with the change of network topology. The above branch outage could result in the overloads of other branches and give rise to the further cascades. At the given cascading step (e.g., the $l$-th cascading step in Fig. \ref{flow}), protective actions (e.g., load shedding and generation dispatch) are taken to prevent the cascading failure. The above cascade process can be described by the state equation in optimal control theory. By designing a cost function to quantify the final disruption level of cascading failures, an optimal control algorithm can be developed to obtain the optimal control inputs, which are exactly the initial contingencies or disturbances that can cause the worst-case cascading failures.

\begin{figure}
\scalebox{0.045}[0.045]{\includegraphics{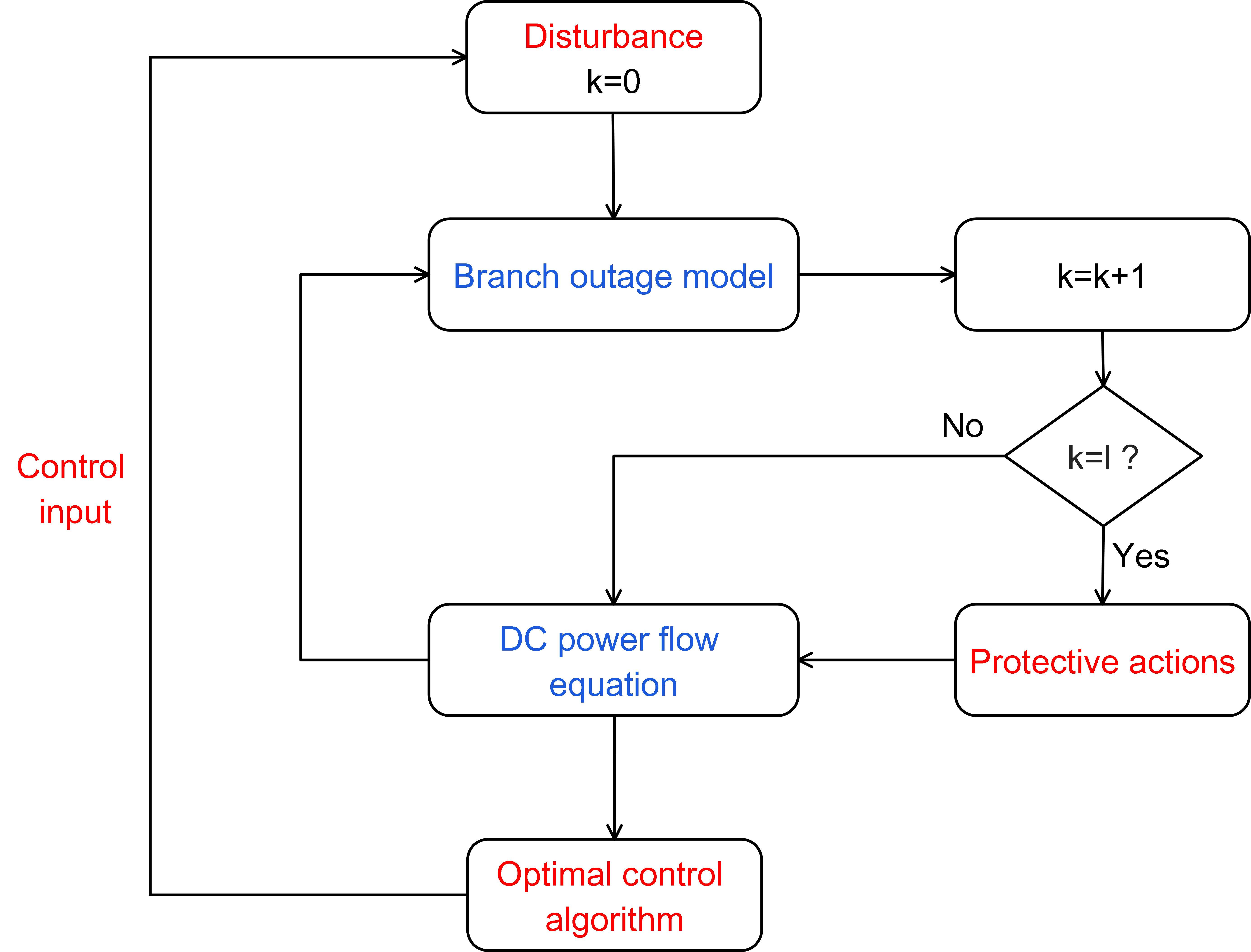}}\centering
\caption{\label{flow} Schematic diagram on the identification of power system cascades with protective actions.}
\end{figure}

By treating the branch impedance or admittance as the state variable, a state equation can be established to characterize the line outage sequence during the power system cascades. The power flow on each branch is obtained by solving the power flow equation. The protective actions during cascades are implemented at the given cascading step by adjusting the injected power on buses. Finally, an optimal control problem is formulated to search for the initiating event of cascading blackouts with the consideration of load shedding and generation dispatch.

\subsection{State equation}
The state equation is used to describe the line outage sequence of power grids during cascades. Suppose the status of branch or transmission line depends on both the power flow on the branch and its threshold. Thus, an approximate function on the branch status
is designed to characterize the operation of circuit breaker in protective relays as follows
$$
g(P_{ij},c_{ij})=\left\{
                \begin{array}{ll}
                  0, & \hbox{$|P_{ij}|\geq \sqrt{c_{ij}^2+\frac{\pi}{2\sigma}}$;} \\
                  1, & \hbox{$|P_{ij}|\leq\sqrt{c_{ij}^2-\frac{\pi}{2\sigma}}$;} \\
                  \frac{1-\sin\sigma (P_{ij}^2-c_{ij}^2)}{2}, & \hbox{otherwise.}
                \end{array}
              \right.
$$
where $\sigma$ is a tunable positive parameter. $P_{ij}$ refers to the power flow on the branch between Bus $i$ and Bus $j$, and $c_{ij}$ denotes its threshold of power flow. $P_{ij}$ can be obtained by solving the power flow equation. Notably, the function $g$ is differentiable with respect to $P_{ij}$, and it approximates to the step function as the parameter $\sigma$ increases. It can reflect the physical characteristic of protective relays while contributing to theoretical analysis on the identification of initial contingencies using optimal control theory.

The line outage sequence can be described by the following state equation (i.e., branch outage model in Figure \ref{flow})
\begin{equation}\label{state_eq}
Y_p^{k+1}=G(P_{ij}^k)\cdot Y_p^{k}+E_{i_k}u_k
\end{equation}
where $Y_p^k=(y^k_{p,1},y^k_{p,2},...,y^k_{p,n})^T\in R^n$ refers to the vector of branch admittance, and the control input or disturbance is denoted by $u_k=(u_{k,1},u_{k,2},...,u_{k,n})^T$. In addition, $G(P_{ij}^k)$ is a diagonal matrix
$$
G(P_{ij}^k)=\left(
              \begin{array}{cccc}
                g(P_{i_1j_1},c_{i_1j_1}) & 0 & . & 0 \\
                0 & g(P_{i_2j_2},c_{i_2j_2}) & . & 0 \\
                . & . & . & . \\
                0 & 0 & . & g(P_{i_nj_n},c_{i_nj_n}) \\
              \end{array}
            \right)
$$
and it describes the status of each branch at the $k$-th cascading step. The matrix $E_{i_k}$ is constructed as
$$
E_{i_k}=diag\left(e^T_{i_k}\right)=diag(\underbrace{0,..,0,1}_{i_k},0,...,0) \in R^{n \times n}
$$
and it enables us to select the $i_k$-th branch for adding the initial disturbances.

\subsection{DC power flow equation}
For high-voltage transmission systems, the DC power flow is a good substitute for the AC power flow \cite{yan15}, and it is computationally efficient and immune to the numerical non-convergence. Thus, the DC power flow equation is employed in this work. The classic DC power flow equation is given by
\begin{equation}\label{pfe}
P_i=\sum_{j=1}^{m}B_{ij}\theta_{ij}=\sum_{j=1}^{m}B_{ij}(\theta_i-\theta_j)
\end{equation}
where $P_i$ and $\theta_i$ refer to the injection power and voltage phase angle of Bus $i$, respectively. $B_{ij}$ represents the mutual susceptance between Bus $i$ and Bus $j$. The symbol $m$ denotes the number of buses in power networks. Equation (\ref{pfe}) can be rewritten in matrix form \cite{stot09}
$$
P=B\theta
$$
where $P=(P_1, P_2,...,P_m)$, $\theta=(\theta_1,\theta_2,...,\theta_m)$ and
$$
B=\left(
   \begin{array}{cccc}
     \sum_{i=2}^{m}B_{1i} & -B_{12} & . & B_{1m} \\
     -B_{21} & \sum_{i=1,i\neq2}^{m}B_{2i} & . & B_{2m} \\
     . & . & . & . \\
     -B_{m1} & -B_{m2} & . & \sum_{i=1}^{m-1}B_{mi} \\
   \end{array}
  \right)
$$
It is worth pointing out that $B$ can be obtained by removing the real part of $Y_b$. In fact, the nodal admittance matrix $Y^k_{b}$ at the $k$-th cascading step can be obtained as
$$
Y^k_{b}=A^Tdiag(Y^k_p)A
$$
where $A$ denotes the branch-bus incidence matrix \cite{stag68}. Thus, the power flow on the branch between Bus $i$ and Bus $j$ can be computed as
$$
P^k_{ij}=e_i^TY_b^ke_j(e_i-e_j)^T(Y_b^k)^{-1^*}P^k,~i,j\in I_{m}=\{1,2,...,m\}
$$
where $e_i$ is the $m$-dimensional unit vector with $1$ at the $i$-th position and $0$ elsewhere. $P^k$ refers to the vector of injected power on buses at the $k$-th cascading step. The symbol $-1^*$ represents a generalized inverse for solving DC power flow equation \cite{zhai17}.

\begin{remark}
During the cascading blackout, power network may be divided into several subnetworks ($i.e.$, islands), which can be identified by analyzing the nodal admittance matrix $Y^k_{b}$. To solve the DC power flow equation, the generator bus connected to the largest generating station is selected as the new slack bus in the subnetwork. And thus the power variation of slack bus accounts for a small percentage of its generating capacity. If there is no generator bus in the subnetwork, the power flow is zero on each branch of this subnetworks.
\end{remark}

\subsection{Protective actions}
If generation dispatch and load shedding are taken into account in the formulation, $P^k$ has to be updated at certain steps of cascading failure. For simplicity, suppose that load shedding and generation control are implemented at the $l$-th cascading step ($1<l<h$).
This implies that $P^k=P^l$ for $k\geq l$. Thus, a nonlinear programming problem can be proposed to allow for load shedding and generation dispatch as follows
\begin{equation}\label{orig}
\begin{split}
    &~~~~~~~~\min_{P^l}\|P^l-P^0\|^2 \\
    &s.~t.~~~~~\underline{P}_i\leq P^l_i \leq \bar{P}_i \\
    &~~~~~~-c_{ij} \leq P^l_{ij} \leq c_{ij} \\
\end{split}
\end{equation}
where $P^l=(P_1^l,P_2^l,...,P_m^l)^T$, and $P^0$ denotes the vector of original injected power on buses. The symbols $\underline{P}_i$ and $\bar{P}_i$ denote the upper and lower bounds of injected power on Bus $i$, respectively. The cost function in (\ref{orig}) quantifies the changes of injected power on buses due to load shedding and generation control. Essentially, the objective of Optimization Problem (\ref{orig}) is to achieve the minimum adjustment of injected power on buses while preventing the further branch outages of power grids.
\begin{remark}
The linear programming formulation can also be adopted to allow for protection actions in power systems \cite{car02}. The proposed approach is also applied to the linear programming formulation. In addition, it also has a chance to be extended for  dealing with protection actions at multiple cascading steps.
\end{remark}

\subsection{Cost function}
Next, a cost function of optimal control problem is presented to quantify the disruption level of power grids at the end of cascading failures. Suppose the cascades come to an end at the $h$-th cascading step, and the cost function is designed as
\begin{equation}\label{cost}
\min_{u_k}J(Y_p^h,u_k)
\end{equation}
with
$$
J(Y_p^h,u_k)=\Gamma(Y_p^h)+\epsilon\sum_{k=0}^{h-1}\frac{\|u_k\|^2}{\max\{0,1-k\}}
$$
and the adjustment of injected power on buses $P^k$ for the system protection is implemented according to the solutions to Optimization Problem (\ref{orig}) at the $l$-th cascading step. The first term in the cost function (i.e., $\Gamma(Y_p^h)$) describes the final status of cascading failures (e.g., network connectivity or power flow), and the second term characterizes the accumulated control energy. In addition, $\epsilon$ is a positive weight, and the symbol $\|\cdot\|$ represents the $2$-norm. Normally, $\epsilon$ is small enough so that more efforts are taken to minimize the first term in the cost function.
\begin{remark}
The second term in the cost function ensures that only initial disturbances can be added on the branch of power transmission system. This is due to $\max\{0,1-k\}=1$ for $k=0$ and $\max\{0,1-k\}=0$ for $k\geq1$. When $\max\{0,1-k\}=0$, the second term goes to the infinity. This implies that it is impossible to minimize the cost function when $k\geq1$.
\end{remark}

\section{Theoretical Results}\label{sec:the}
This section provides theoretical results on the identification of initial contingencies on branches that can cause the catastrophic cascading failures of power grids. First of all, an equivalent condition is presented for Problem (\ref{orig}).
\begin{lemma}\label{lem1}
The optimal solutions of Optimization Problem (\ref{orig}) are equivalent to the solutions of the following system of algebraic equations
\begin{equation}\label{kkt}
\begin{split}
    &2(P^l-P^0)+\bar{\mu}-\underline{\mu}+\sum_{(i,j)\in\Omega} (\bar{\lambda}_{ij}-\underline{\lambda}_{ij})e_i^TY_b^le_j(Y_b^l)^{-1^*}e_{ij}=\bf{0} \\
    &P^l_i-\bar{P}_i+\bar{x}^2_i=0, ~~~\quad \bar{\mu}_i(P^l_i-\bar{P}_i)=0,~~\quad \bar{\mu}_i-\bar{z}^2_i=0 \\
    &\underline{P}_i-P^l_i+\underline{x}^2_i=0,~~~\quad \underline{\mu}_i(P^l_i-\underline{P}_i)=0,~\quad \underline{\mu}_i-\underline{z}^2_i=0 \\
    &P^l_{ij}-c_{ij}+\bar{y}^2_{ij}=0, \quad (P^l_{ij}-c_{ij})\bar{\lambda}_{ij}=0, \quad \bar{\lambda}_{ij}-\bar{w}^2_{ij}=0\\
    &P^l_{ij}+c_{ij}-\underline{y}^2_{ij}=0, \quad (P^l_{ij}+c_{ij})\underline{\lambda}_{ij}=0, \quad \underline{\lambda}_{ij}-\underline{w}^2_{ij}=0 \\
\end{split}
\end{equation}
where $Y^l_{b}=A^Tdiag(Y^l_p)A$, $e_{ij}=e_i-e_j$, $i\in I_m$, $(i,j)\in\Omega$, $\bar{\mu}=(\bar{\mu}_1,\bar{\mu}_2,...,\bar{\mu}_m)^T$ and $\underline{\mu}=(\underline{\mu}_1,\underline{\mu}_2,...,\underline{\mu}_m)^T$. And
the symbol $\Omega$ denotes the set of branches in power systems with the cardinality $|\Omega|=n$ (i.e., the number of branches).
\end{lemma}

\begin{proof}
The inequality constraints in the optimization problem (\ref{orig}) can be converted into the equality constraints by introducing the unknown variables $\bar{x}_i$, $\underline{x}_i$, $\bar{y}_{ij}$ and $\underline{y}_{ij}$ as follows.
\begin{equation}\label{equal}
\begin{split}
    &~~~~~~~~\min_{P^l}\|P^l-P^0\|^2 \\
    &s.~t. ~~~~P^l_i-\bar{P}_i+\bar{x}^2_i=0 \\
    &~~~~~~~~~\underline{P}_i-P^l_i+\underline{x}^2_i=0 \\
    &~~~~~~~~~P^l_{ij}-c_{ij}+\bar{y}^2_{ij}=0 \\
    &~~~~~~~~~P^l_{ij}+c_{ij}-\underline{y}^2_{ij}=0
\end{split}
\end{equation}
According to the  Karush-Kuhn-Tucker (KKT) conditions, the necessary condition for the solutions to Optimization Problem (\ref{equal}) can be obtained as follows. Specifically, we have
$$
\nabla\|P^l-P^0\|^2+\bar{\mu}-\underline{\mu}+\sum_{(i,j)\in\Omega}(\bar{\lambda}_{ij}-\underline{\lambda}_{ij})e_i^TY_b^le_j(Y_b^l)^{-1^*}e_{ij}=\bf{0} $$
with $\nabla\|P^l-P^0\|^2=2(P^l-P^0)$ for the stationarity condition and the constraints
\begin{equation*}
\begin{split}
    &P^l_i-\bar{P}_i+\bar{x}^2_i=0, \quad  P^l_{ij}-c_{ij}+\bar{y}^2_{ij}=0 \\
    &\underline{P}_i-P^l_i+\underline{x}^2_i=0, \quad P^l_{ij}+c_{ij}-\underline{y}^2_{ij}=0
\end{split}
\end{equation*}
The conditions for complementary slackness are given by
\begin{equation*}
\begin{split}
    &\bar{\mu}_i(P^l_i-\bar{P}_i)=0, \quad (P^l_{ij}-c_{ij})\bar{\lambda}_{ij}=0 \\
    &\underline{\mu}_i(P^l_i-\underline{P}_i)=0, \quad (P^l_{ij}+c_{ij})\underline{\lambda}_{ij}=0
\end{split}
\end{equation*}
Moreover, the dual feasibility can be described by
\begin{equation*}
\begin{split}
    &\bar{\mu}_i-\bar{z}^2_i=0, \quad \bar{\lambda}_{ij}-\bar{w}^2_{ij}=0 \\
    &\underline{\mu}_i-\underline{z}^2_i=0 \quad \underline{\lambda}_{ij}-\underline{w}^2_{ij}=0
\end{split}
\end{equation*}
where $e_{ij}=e_i-e_j$, $i\in I_m$, $(i,j)\in\Omega$, $\bar{\mu}=(\bar{\mu}_1,\bar{\mu}_2,...,\bar{\mu}_m)^T$ and $\underline{\mu}=(\underline{\mu}_1,\underline{\mu}_2,...,\underline{\mu}_m)^T$.
Since the cost function in Optimization Problem (\ref{orig}) is a convex function and the inequality constraints are affine, the above necessary conditions are also sufficient for optimality. This implies that the solutions to Optimization Problem (\ref{orig}) are equivalent to solutions to the system of algebraic equations (\ref{equal}). The proof of this lemma is thus completed.
\end{proof}

\begin{remark}
As we can see, System (\ref{kkt}) is composed of $(7m+6n)$ equations and $(7m+6n)$ additional unknown variables (i.e., $P^l$, $\underline{w}_{ij}$, $\bar{w}_{ij}$, $\underline{\lambda}_{ij}$, $\bar{\lambda}_{ij}$, $\underline{y}_{ij}$, $\bar{y}_{ij}$, $\underline{x}_{i}$, $\bar{x}_{i}$, $\underline{z}_{i}$, $\bar{z}_{i}$, $\underline{\mu}_{i}$, $\bar{\mu}_{i}$). Note that $Y^l_b$ contains the existing unknown variables $Y_p^l$ in the state equation (\ref{state_eq}).
\end{remark}

The equivalent conditions in Lemma \ref{lem1} allows to obtain the necessary conditions for the optimal control problem (\ref{cost}) as follows

\begin{theorem}
The necessary conditions for the optimal control problem (\ref{cost}) with protective actions according to (\ref{orig}) are given by solving the system of algebraic equations as follows
\begin{equation}\label{con_sys}
\left\{
  \begin{array}{ll}
    Y_p^{k+1}-G(P^k_{ij})Y_p^k-\frac{\max\{0,1-k\}}{2\epsilon}E_{i_k}\prod_{s=0}^{h-k-2}\frac{\partial Y_p^{h-s}}{\partial Y_p^{h-s-1}}\mathbf{1}_n=\mathbf{0} \\
    \mathcal{F}(Y_p^l, P^l, \underline{w}_{ij}, \bar{w}_{ij}, \underline{\lambda}_{ij}, \bar{\lambda}_{ij}, \underline{y}_{ij}, \bar{y}_{ij}, \underline{x}_{i}, \bar{x}_{i}, \underline{z}_{i}, \bar{z}_{i}, \underline{\mu}_{i}, \bar{\mu}_{i})=\mathbf{0}
  \end{array}
\right.
\end{equation}
where the second equation represents System (\ref{kkt}), and the optimal adjustment of injected power on buses for protective actions satisfies $P^k=P^l$ for $k\geq l$, and $P^l$ is the solution to Problem (\ref{orig}). In addition, the optimal control input is given by
\begin{equation}\label{con_input}
u_k=\frac{\max\{0,1-k\}}{2\epsilon}E_{i_k}\prod_{s=0}^{h-k-2}\frac{\partial Y_p^{h-s}}{\partial Y_p^{h-s-1}}\mathbf{1}_n, \quad k\in I_{h-1}
\end{equation}
\end{theorem}

\begin{proof}
With Pontryagin's maximum principle in optimal control theory for the discrete-time system \cite{fra95}, the necessary conditions for the optimal control problem (\ref{cost}) can be determined as
\begin{equation*}
Y_p^{k+1}=G(P_{ij}^k)\cdot Y_p^{k}+E_{i_k}u_k
\end{equation*}

\begin{equation*}
\left(\frac{\partial Y_p^{k+1}}{\partial u_k}\right)^T\lambda_{k+1}+\frac{\epsilon }{\max\{0,1-k\}}\cdot\frac{\partial\|u_k\|^2}{\partial u_k}=0
\end{equation*}

\begin{equation*}
\lambda_k=\left(\frac{\partial Y_p^{k+1}}{\partial Y_p^k}\right)^T\lambda_{k+1}+\frac{\epsilon }{\max\{0,1-k\}}\cdot\frac{\partial\|u_k\|^2}{\partial Y_p^k}
\end{equation*}

\begin{equation*}
\frac{\partial \mathrm{T}(Y_p^h)}{\partial Y_p^h}-\lambda_h=\mathbf{0}
\end{equation*}
where $\mathbf{0}=(0,0,...,0)^T \in R^n$. By reorganizing the above equations,
we can obtain the optimal control input
$$
u_k=\frac{\max\{0,1-k\}}{2\epsilon}E_{i_k}\prod_{s=0}^{h-k-2}\frac{\partial Y_p^{h-s}}{\partial Y_p^{h-s-1}}\mathbf{1}_n, \quad k\in I_{h-1}
$$
and the system of algebraic equations
\begin{equation}\label{s1}
Y_p^{k+1}-G(P^k_{ij})Y_p^k-\frac{\max\{0,1-k\}}{2\epsilon}E_{i_k}\prod_{s=0}^{h-k-2}\frac{\partial Y_p^{h-s}}{\partial Y_p^{h-s-1}}\mathbf{1}_n=\mathbf{0}
\end{equation}
where the vector $P^l$ is determined by the solutions to System (\ref{kkt}), which can be rewritten as
\begin{equation}\label{s2}
\mathcal{F}(Y_p^l, P^l, \underline{w}_{ij}, \bar{w}_{ij}, \underline{\lambda}_{ij}, \bar{\lambda}_{ij}, \underline{y}_{ij}, \bar{y}_{ij}, \underline{x}_{i}, \bar{x}_{i}, \underline{z}_{i}, \bar{z}_{i}, \underline{\mu}_{i}, \bar{\mu}_{i})=\mathbf{0}
\end{equation}
for ease of notation according to Lemma \ref{lem1}. By combining Equations (\ref{s1}) and (\ref{s2}), an extended system of algebraic equations is obtained with $7m+(6+h)n$ equations and $7m+(6+h)n$ unknown variables. By substituting the solutions of this extended system into (\ref{con_input}), we can identify the most disruptive disturbances for the cascades of power grids with protective actions at the given cascading step.  This completes the proof
\end{proof}

\begin{remark}
The extended system of algebraic equations (\ref{con_sys}) can be solved using the numerical solver in numerical-analysis softwares. In this way, load shedding and generation dispatch can be taken into account in the proposed optimal control formulation of identifying the worst-case cascading failures.
\end{remark}

\begin{remark}
Since the solutions to System (\ref{con_sys}) can only provide necessary conditions for the optimal control problem (\ref{cost}), extensive numerical simulations have to be conducted to search for the feasible solutions in practice. By comparing the values of cost function with the identified initial disturbances (i.e., control inputs), it is expected to obtain the optimal or suboptimal solutions
to Problem (\ref{cost}) if numerical simulations are repeated for a sufficiently large number of times.
\end{remark}

\section{Simulation and Validation}\label{sec:sim}
To demonstrate the effectiveness of the proposed identification approach based on optimal control, numerical simulations are conducted on the IEEE RTS 24 Bus System to determine the initial disruptive disturbances on each branch (see Fig.~\ref{b24}). In addition, we make a comparison on simulation results between the cases with and without protective actions (i.e., load shedding and generation dispatch), respectively.

\subsection{Parameter setting}
\begin{figure}
\scalebox{0.45}[0.45]{\includegraphics{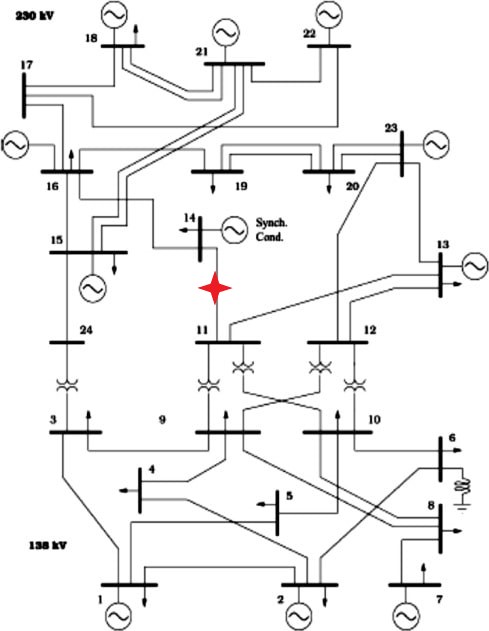}}\centering
\caption{\label{b24} IEEE RTS 24 bus system with initial disturbances on branches (red star)}
\end{figure}

Per unit values are adopted with the base value of power $100$ MVA. Moreover, the solver ``fsolve" in Matlab is employed to solve the integrated system of algebraic equations (\ref{con_sys}). Other parameters are given as follows: $\sigma=5\times 10^4$, $\epsilon=10^{-4}$, $h=10$, and $\Gamma(Y_p^h)=\|Y_p^h\|^2/2$ in the cost function (\ref{cost}). For each branch, numerical simulations are carried out to solve System (\ref{con_sys}) for $10$ times, and the worst-case disturbances are selected by comparing the values of cost function. A performance index $\gamma$ is defined to quantify the disruption level of cascades as follows
$$
\gamma=\frac{J(Y_p^h,u)}{J(Y_p^h,0)}.
$$
Intuitively, the index $\gamma$ is the ratio between the final cost of power grids with the control input $u$ and that without any control inputs, and a smaller $\gamma$ indicates a worse cascade of power systems. The power flow threshold on each branch is $10\%$ larger than the normal power flow on the corresponding branch of power systems without any disturbances.

\subsection{Validation and comparison}
In the simulations, generation dispatch and load shedding are implemented at the 4-th cascading step~(i.e., $l=4$) according to solutions of Optimization Problem (\ref{orig}). Figure \ref{comp} presents the initial disturbances (i.e., control inputs) on each branch identified by the proposed optimal control approach and the resulting normalized costs (i.e., the index $\gamma$). Specifically, the blue bars denote the control inputs and normalized costs without generation dispatch and load shedding, while the green bars represent those with generation dispatch and load shedding. As we can observe in the upper panel of Figure~\ref{comp}, the height of green bar is not smaller than that of blue bar for each branch. This indicates that the larger initial disturbances~(i.e., the magnitude of control input)~are required to trigger the worst-case cascades of power grids with generation dispatch and load shedding compared to those without generation dispatch and load shedding. This demonstrates that protection actions are able to effectively enhance the robustness of power systems and relieve the final disruption level after malicious disturbances. It is worth pointing out that the branches with equal heights of blue bar and green bar (e.g., Branch 1, Branch 2, Branch 3, Branch 4, Branch 5, Branch 6, etc) are directly severed by the initial disturbances (i.e., control inputs). The lower panel of Figure~\ref{comp} demonstrates that the cascades with generation dispatch and load shedding are less disruptive on the whole except for the cascades triggered by the initial disturbances on Branch 28, Branch 32 and Branch 33. This is because the optimal solutions for generation dispatch and load shedding (i.e., solutions to the optimization problem (\ref{orig})) are not obtained by the numeric solver at the specified cascading step. Thus, the adjustment of injected power on buses (i.e., generation dispatch and load shedding) actually deteriorates branch overloads and results in the worse disruptions of cascades in the end.
\begin{figure}
\scalebox{0.68}[0.68]{\includegraphics{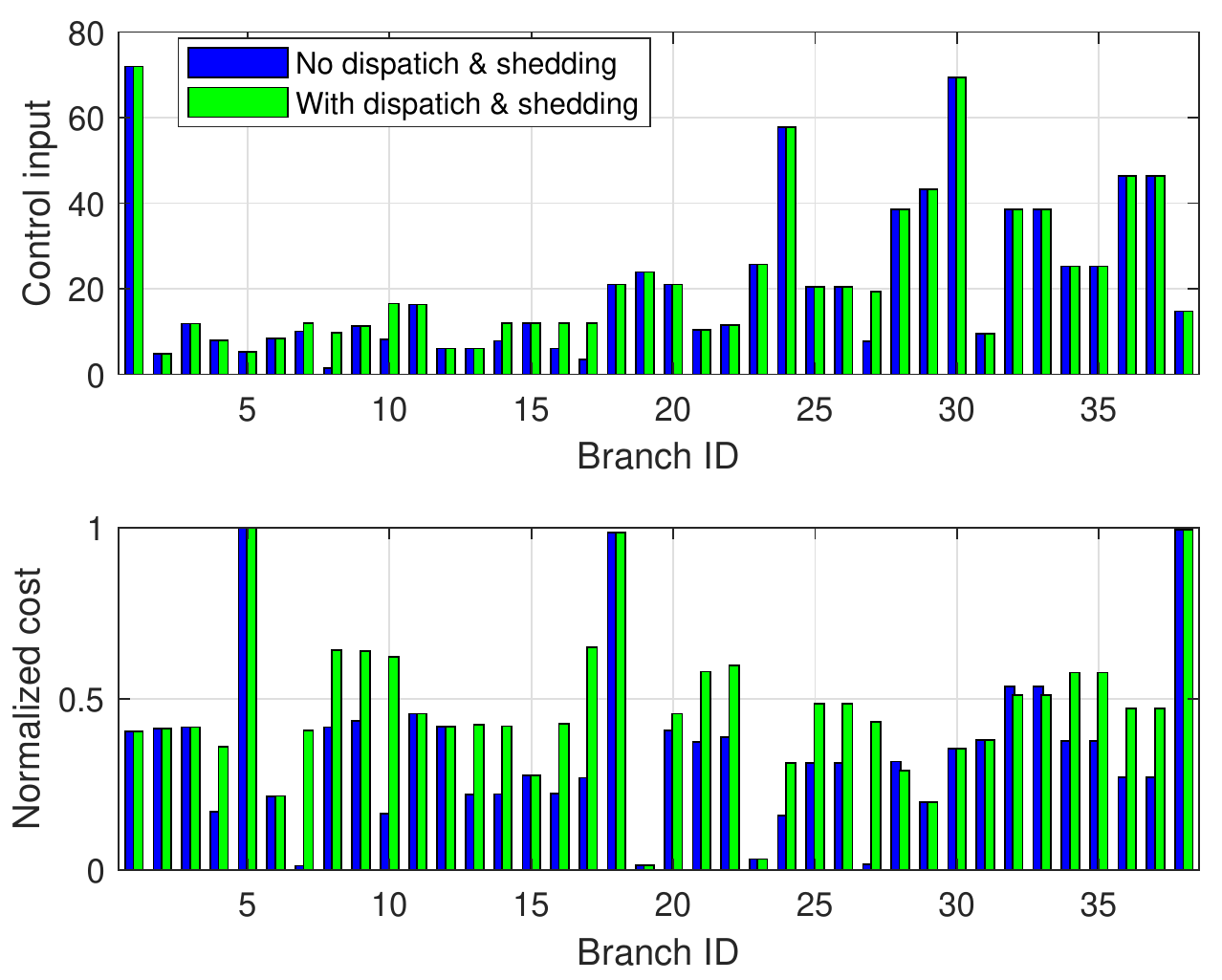}}\centering
\caption{\label{comp} Control inputs and normalized costs of the cascades with and without generation dispatch and load shedding on the IEEE RTS 24 Bus System.}
\end{figure}

\section{Conclusions}\label{sec:con}
In this paper, we investigated the problem of identifying the initial contingencies that result in the worst-case cascading failures of power grids. In particular, power grids are equipped with protection devices to prevent the cascading blackouts by load shedding or generation dispatch in time.  Moreover, a theoretical framework was proposed to allow for both the identification of the most disruptive disturbances and the optimal adjustment of injected power buses for the protection of power girds. Numerical simulations were conducted to better understand the effect of protection actions on the identification of initial malicious contingencies and the final disruptions of cascading failures. In this work, the deterministic cascading failure paths are taken into account for the identification of initial disruptive disturbances and the implementation of protective actions. In practice, power grids are subject to uncertainties (e.g., hidden failure, device aging, human errors, etc). Therefore, we will consider the contingency identification of power grids with uncertain cascading failure paths in the next step.

\section*{Acknowledgment}
This work is partially supported by the Future Resilience System Project at the Singapore-ETH Centre (SEC), which is funded by the National Research Foundation of Singapore (NRF) under its Campus for Research Excellence and Technological Enterprise (CREATE) program. It is also supported by Ministry of Education of Singapore under Contract MOE2016-T2-1-119.

\end{document}